\documentclass[aps,prx,reprint,twocolumn,superscriptaddress,nofootinbib]{revtex4-2}

\usepackage{amsthm,amsmath,amssymb,amsfonts,amsthm}
\usepackage{tikz}
\usetikzlibrary{positioning, arrows.meta, shapes.geometric, calc, fit,backgrounds}

\usepackage{graphicx}
\usepackage{hyperref}
\usepackage{algorithm}
\usepackage{algpseudocode}
\usepackage{listings}
\usepackage{xcolor}
\usepackage{booktabs}

\newtheorem{theorem}{Theorem}
\newtheorem{definition}[theorem]{Definition}

\newtheorem{proposition}[theorem]{Proposition}

\hypersetup{
  linkcolor  = red,
  citecolor  = blue,
  urlcolor   = purple,
  colorlinks = true,
}

\def\nn{\nonumber}
\def\q{\quad}
\def\qq{\qquad}

\def\sixj{ \{ { 6j} \} }

\def\su2{{\rm SU(2)}}
\def\Q{\mathbb Q}
\def\Z{\mathbb Z}
\def\R{\mathbb R}

\def\C{\mathbb C}

\begin{document}

\author{Seth K. Asante}
\email{seth.asante@unb.ca}
\affiliation{Department of Mathematics and Statistics,  University of New Brunswick, Fredericton, NB, E3B 5A3, Canada} 

\title{Deferred Cyclotomic Representation for Stable and Exact Evaluation of q-Hypergeometric Series}


\begin{abstract}

We introduce a cyclotomic representation for finite $q$-hypergeometric series and $q$-deformed amplitudes that separates algebraic structure from evaluation. By expressing each summand in a sparse exponent basis over irreducible cyclotomic polynomials, all products and ratios of quantum factorials reduce to integer vector arithmetic. This ensures that cancellations between numerator and denominator are resolved exactly prior to any evaluation. This formulation yields the \emph{deferred cyclotomic representation} (DCR), a parameter-independent combinatorial object of the series, from which evaluation in any target field is realized as a ring homomorphism.

For quantum recoupling coefficients, we demonstrate that this framework achieves linear memory scaling in the compilation phase, eliminates intermediate expression swell in exact arithmetic, and substantially extends the range of reliable double-precision computation by reducing cancellation-induced error amplification. Beyond its computational advantages, the DCR provides a unified perspective on $q$-deformed amplitudes. Structural properties like admissibility at roots of unity, and the classical limit all emerge as intrinsic properties of a single underlying combinatorial object.

\end{abstract}

\maketitle

\section{Introduction}
\label{sec:introduction}

Finite alternating $q$-hypergeometric series arise in a wide range of applications including representation theory, quantum topology, and mathematical physics. They encode the algebraic structure of quantum groups \cite{Drinfeld1985Hopf,kassel1995quantum}, appear in the construction of topological quantum field theories and knot invariants \cite{witten1989quantum,Turaev1992,kauffman1994temperley}, and govern recoupling coefficients such as quantum $6j$-symbols \cite{kirillov1989reps,chari1995guide}. These same series are used to define fundamental transition amplitudes in state-sum models of three-dimensional quantum gravity \cite{Ponzano1968,rovelli2004quantum,Oriti2011,Barrett:1993ab,Major:1995yz,Dupuis:2020ndx}, as well as in tensor network formulations of topological phases of matter \cite{Kitaev2003,etingof2015tensor,wen2004quantum}.

Despite their fundamental importance, the evaluation of $q$-hypergeometric series remains computationally challenging. Exact symbolic approaches suffer from severe intermediate \emph{expression swell} \cite{Geddes1992, vonZurGathen2013}, since the rational functions grow rapidly in size with the summation range. Numerical evaluation, on the other hand, is limited by \emph{catastrophic cancellation}, where alternating sums of exponentially large terms lead to substantial loss of precision  \cite{Higham2002, goldberg1991,ascher2011numerical}. These difficulties are particularly acute at roots of unity and in semiclassical regimes, where the \emph{condition number} of the summation grows rapidly with the parameters.

These limitations are not only due to algorithmic implementations, but also reflect the algebraic representation in which the amplitudes are expressed. In conventional approaches, $q$-hypergeometric series are expressed as dense rational or polynomial functions of the deformation parameter $q$. This representation obscures the underlying multiplicative structure of quantum factorials and forces cancellations to occur only after numerical or symbolic expansion. As a result, both exact and floating-point computations are exposed to unnecessary intermediate complexity.

In this work, we introduce a representation that is naturally adapted to multiplicative structure of finite $q$-hypergeometric series, and separates its algebraic structure from evaluation. The central idea is to express quantum integers and factorials through their cyclotomic factorization and to encode the resulting expressions as products of the irreducible cyclotomic generators $\{q,\Phi_d(q^2)\}_{d \ge 2}$. We refer to this construction as the \emph{deferred cyclotomic representation} (DCR). Within this representation, each $q$-hypergeometric summand is encoded as a sparse vector of integers corresponding to exponents of the cyclotomic basis. As a result, all multiplicative operations reduce to integer arithmetic.

The key conceptual shift is that the dependence on the deformation parameter $q$ is removed from the algebraic representation and reintroduced only at the level of evaluation. In this framework, a $q$-hypergeometric amplitude is not a function to be recomputed for each value of $q$, but the evaluation of a fixed combinatorial object under a family of projection maps. This perspective unifies exact symbolic computation, floating-point evaluation, and asymptotic limits within a single representation.

\medskip

\noindent
\textbf{Main contributions.}
We summarize the main contributions of this work as follows:
\begin{itemize}
\item[(1)] \emph{Cyclotomic representation.} We show that finite $q$-hypergeometric amplitudes admit a representation in a sparse cyclotomic exponent basis, revealing the multiplicative structure of quantum factorials.

\item[(2)] \emph{Deferred evaluation.} We introduce the deferred cyclotomic representation (DCR), which separates a one-time combinatorial compilation from subsequent evaluation via projection into a target field.

\item[(3)] \emph{Efficient construction.} We show that the construction of the DCR scales near-linearly in the summation length and avoids the expression swell inherent in polynomial representations.

\item[(4)] \emph{Numerical stability.} We show that the DCR performs algebraic cancellations prior to evaluation. This reduces the dynamic range of intermediate quantities and mitigate cancellation-driven loss of precision in floating-point arithmetic.

\item[(5)] \emph{Structural unification.} We demonstrate that admissibility at roots of unity, $q$-deformation, and classical limits manifest as intrinsic properties of a single DCR combinatorial object.
\end{itemize}

From a computational standpoint, this separation yields several advantages. By encoding all multiplicative structure at the level of integer exponents, the DCR eliminates intermediate polynomial growth and reduces memory usage in exact arithmetic. At the same time, it acts as an \emph{algebraic preconditioner} for numerical evaluation, performing exact cancellations prior to floating-point computation and thereby reducing error amplification. Once constructed, the representation can be reused across multiple parameter regimes, enabling efficient amortized computation over large parameter spaces.

Beyond its computational advantages, the DCR also provides a different perspective on the algebraic organization of $q$-deformed amplitudes. Because the exponent vectors are independent of both the deformation parameter and the target field, phenomena that are usually treated separately (including root-of-unity admissibility, classical limits, and evaluation in different arithmetic domains), are realized as different projections of a common combinatorial object. In this sense, the DCR separates combinatorial structure from arithmetic realization. This viewpoint suggests new connections between computational representations, cyclotomic arithmetic, and quantum topology, which we discuss in Section~\ref{sec:structure}.

More broadly, the DCR should not be viewed simply as an algorithm for evaluating a particular family of $q$-hypergeometric series. Rather, it provides a representation in which the combinatorial structure of the amplitudes is separated from their arithmetic realization. This separation simultaneously yields computational advantages, clarifies the algebraic organization of $q$-deformed amplitudes, and suggests new connections between computational mathematics, cyclotomic arithmetic, and quantum topology. Several of these structural consequences motivate directions for future investigation.

We develop this framework in detail and analyze its properties both theoretically and empirically. Using the quantum $6j$-symbol as a benchmark, we demonstrate linear memory scaling in the construction phase, improved numerical stability in floating-point arithmetic, and efficient amortized evaluation across continuous ranges of $q$. The $6j$-symbol provides a canonical test case in which factorial growth, oscillatory summation, and semiclassical behaviour are all present, making it a controlled and representative benchmark for both symbolic and numerical performance. Moreover, because the quantum $6j$-symbol forms the elementary building block of many state-sum and tensor-network constructions, improvements at this level directly benefit larger computational frameworks.

\medskip

The remainder of this paper is organised as follows. Section~\ref{sec:background} reviews the structure of $q$-hypergeometric series and the limitations of conventional evaluation methods. Section~\ref{sec:cyclo_architecture} introduces the cyclotomic representation and the deferred evaluation framework. Section~\ref{sec:universal_projections} defines projection maps and describes their implementation. Section~\ref{sec:structure} develops structural and algebraic
consequences of the representation. Section~\ref{sec:complexity_error} analyses computational complexity and numerical stability. Section~\ref{sec:numerical_results} presents empirical benchmarks. Section~\ref{sec:conclusion} states the conclusions and directions for future work.

\section{Preliminaries and Limitations of Polynomial Representations}
\label{sec:background}

The evaluation of $q$-hypergeometric series is often limited by intrinsic computational difficulties rather than implementation details. In particular, when evaluated at roots of unity, these series combine rapidly growing factorial terms with strongly oscillatory summation patterns. This interplay leads simultaneously to severe numerical cancellation in floating-point arithmetic and substantial intermediate expression growth in exact symbolic computation. As a result, standard ``eager'' evaluation strategies, whether numerical or symbolic, become unreliable or prohibitively expensive as problem sizes increase. 

These computational difficulties are not merely technical. They reflect a deeper representation-level mismatch. In conventional approaches, $q$-hypergeometric series are expressed as dense polynomial or rational functions in the deformation parameter $q$. However, this representation is not intrinsic to the multiplicative structure of quantum integers and factorials. In particular, it obscures their underlying cyclotomic factorization and forces cancellations to occur only after full expansion.

This misalignment has two immediate consequences. First, both symbolic and numerical methods are forced to operate on representations that artificially inflate intermediate complexity. Second, the cost of evaluation becomes dominated not by the combinatorics of the sum itself, but by the overhead introduced by the chosen algebraic form. From this perspective, the primary obstacle is not algorithmic but structural. The polynomial representation fails to respect the natural algebraic organization of the series.

In this section, we formalize the structure of $q$-hypergeometric series, introduce the quantum $6j$-symbol as a canonical example, and analyze the limitations of conventional evaluation strategies. This analysis will motivate the need for a representation that preserves multiplicative structure and enables cancellation prior to expansion.

\subsection{Algebraic Structure of $q$-Hypergeometric Series}
\label{subsec:q_hypergeometric_structure}

In the representation theory of quantum groups and in quantum topology \cite{ReshetikhinTuraev1991,chari1995guide}, the basic building blocks of amplitudes are the quantum integers
\begin{equation}
[n]_q = \frac{q^n - q^{-n}}{q - q^{-1}}, \qq n \in \Z,
\end{equation}
and the associated quantum factorials $[n]_q! = \prod_{m=1}^n [m]_q$, with $[0]_q! = 1$, where $q$ is treated as a formal parameter. 

A broad class of quantities of interest and invariants can be written as finite sums of the form
\begin{equation}\label{eq:generic_q_series}
S(q) = \sum_{z = z_{\min}}^{z_{\max}} (-1)^z \, q^{f(z)} 
\frac{\prod_i [A_i(z)]_q!}{\prod_j [B_j(z)]_q!},
\end{equation}
where $A_i(z)$ and $B_j(z)$ are affine functions of the summation variable $z$ and of external parameters (e.g. angular momentum spins $j$), and $f(z)$ is a linear or quadratic phase. The summation bounds $z_{\min}$ and $z_{\max}$ are determined by admissibility conditions ensuring  that all factorial arguments are non-negative.

Such expressions arise in recoupling theory for $U_q(\mathfrak{sl}_2)$ and in the evaluation of state-sum invariants in topological quantum field theory \cite{kirillov1989reps, kauffman1994temperley,Turaev1992}. Even in this abstract form, two key structural features are apparent: (i) ratios of factorials produce rapidly varying magnitudes across the summation range, and (ii) the alternating sign $(-1)^z$ leads to oscillatory interference. These features are intrinsic and persist across essentially all non-trivial examples.

\subsection{The Quantum $6j$-Symbol}
\label{subsec:6j}

A canonical example of the series Eq.~\eqref{eq:generic_q_series} is the quantum $6j$-symbol (or $q$-deformed Racah--Wigner coefficient), which governs recoupling transformations in $U_q(\mathfrak{sl}_2)$ representation theory \cite{kauffman1994temperley, chari1995guide}. Beyond its algebraic significance, it serves as the fundamental building block of state sum models in quantum topology and quantum gravity \cite{Turaev1992,varshalovich1988quantum}.

From a computational standpoint, the quantum $6j$-symbol provides a minimal nontrivial setting in which the key computational challenges of evaluating oscillatory $q$-hypergeometric series become manifest. These include factorial growth, alternating sums, and severe cancellation between large intermediate contributions.

For admissible spins $\{j_1,\dots,j_6\}$ associated with a geometric tetrahedron, the quantum $6j$-symbol is defined by the Racah-type expression \cite{kirillov1989reps,kauffman1994temperley}
\begin{align}
\begin{Bmatrix} 
j_1 & j_2 & j_3 \\  j_4 & j_5 & j_6  \end{Bmatrix}_q &= 
\Delta_{j_1 j_2 j_3} \, \Delta_{j_1 j_5 j_6}\, \Delta_{j_2 j_4 j_6} \,\Delta_{j_3 j_4 j_5} \times \nn \\ 
&\sum_{z = z_{\min}}^{z_{\max}} \frac{(-1)^z [z+1]_q!}{\prod_{i=1}^4 [z - a_i]_q! \prod_{y=1}^3 [b_y - z]_q!}.
\label{eq:racah_6j}
\end{align}
The summation limits $z_{\min} = \max_i(a_i)$ and $z_{\max} = \min_y(b_y)$ are governed by the linear combinations 
\begin{align}
a_1 &= j_1+j_2+j_3, \q b_1 = j_1+j_2+j_4+j_5, \nn \\
a_2 &= j_1+j_5+j_6, \q b_2 = j_1+j_3+j_4+j_6, \nn \\
a_3 &= j_2+j_4+j_6, \q b_3 = j_2+j_3+j_5+j_6. \q \\
a_4 &= j_3+j_4+j_5, \nn 
\end{align}
The coefficients $\Delta_{abc}$ represent the geometric triangle coefficients (the radical prefactors), and are given by
\begin{equation}\label{eq:tri_coeff}
\Delta_{abc} = \sqrt{ \frac{[a+b-c]_q! [a-b+c]_q! [-a+b+c]_q!}{[a+b+c+1]_q!} },     
\end{equation}
where the triad $(a,b,c)$ satisfies the standard triangle inequalities ($|a - b| \leq c \leq a+ b$) and the integer sum condition ($a+b+c \in \Z$).

At roots of unity,
\begin{equation}\label{eq:qint_h}
q = \exp\left(\frac{i\pi}{h}\right), \q \text{where} \q h = k + 2.
\end{equation}
the theory becomes finite (a modular tensor category associated with $\mathrm{SU}(2)_k$, where $k\in \Z$), with quantum integers reducing to trigonometric functions and admissibility imposing a hard cutoff. Despite this finiteness, the computational difficulties persist and, in fact, become more pronounced due to enhanced oscillatory behavior.

\subsection{Cancellation and Numerical Conditioning}
\label{subsec:cancellation}

Let $T_z$ denote the magnitude of the $z$-th summand in Eq.~\eqref{eq:generic_q_series}, so that $S = \sum_z (-1)^z T_z$. A central difficulty arises from the large disparity between the scale of individual terms and the final sum.

A useful measure of this effect is the \emph{condition number} defined by
\begin{equation}\label{eq:condition_num}
\kappa = \frac{ \sum_z|T_z|}{|S|},
\end{equation}
which captures the scale separation between individual terms and the final sum. In floating-point arithmetic, the number of digits lost due to cancellation scales as $\log_{10}(\kappa)$ \cite{Higham2002}. 

In the semiclassical regime of the $6j$-symbol, one has $|S| \sim j^{-3/2}$ \cite{Ponzano1968,Roberts1999,Taylor2005}, while individual terms grow rapidly with $j$. As a result, $\kappa$ increases super-polynomially, leading to catastrophic cancellation. Crucially, this instability is not a numerical artifact but an intrinsic property of the series. Even perfectly implemented floating-point algorithms cannot avoid precision loss once $\kappa$ becomes large. To be precise, Table \ref{tab:cancellation} tracks the algebraic growth of the quantum $6j$-symbol Eq.~\eqref{eq:racah_6j} for equal spins $j_i = j$ at level $k = 4j$. At $j = 100$, the precision loss of $\Delta_{\text{loss}} = 13.6$ pushes 64-bit hardware to the extreme. By $j=200$, the maximum term reaches $10^{24}$, requiring a minimum of 28 digits of precision to distinguish the signal from numerical noise.

\begin{table}[htpb]
\centering
\caption{Condition number proxy and precision loss $\Delta_{\rm loss} = \log_{10}(\max_z|T_z|/|S|)$ for symmetric $6j$-symbols ($j_i = j$) at level $k = 4j$. IEEE~754 double precision provides $15.9$ significant digits. }
\begin{tabular}{@{}c@{\hspace{6pt}}c@{\hspace{7pt}}c@{\hspace{10pt}}c@{\hspace{12pt}}c@{}}
\toprule
\textbf{Spin $j$} & \textbf{Level $k$} & \textbf{$\max_z |T_z|$} & \textbf{$|S|$} & \textbf{$\Delta_{\text{loss}}$ } \\
\midrule
50  & 200  & $6.03 \times 10^{3\hphantom{00}}$  & $2.19 \times 10^{-3}$ & 6.4  \\
100 & 400  & $2.96 \times 10^{10\hphantom{0}}$ & $7.80 \times 10^{-4}$ & 13.6 \\
200 & 800  & $2.82 \times 10^{24\hphantom{0}}$ & $2.77 \times 10^{-4}$ & 28.0 \\
300 & 1200 & $4.74 \times 10^{38\hphantom{0}}$ & $1.51 \times 10^{-4}$ & 42.5 \\
400 & 1600 & $1.01 \times 10^{53\hphantom{0}}$ & $9.80 \times 10^{-5}$ & 57.0 \\
\bottomrule
\end{tabular}
\label{tab:cancellation}
\end{table}

Standard remedies, such as logarithmic rescaling via the Log-Sum-Exp (LSE) transformation \cite{Blanchard2012}, mitigate overflow but do not resolve cancellation. High-precision arithmetic (using \texttt{BigFloat}) delays failure but introduces significant computational overhead and does not address the underlying structural cause.

\subsection{Complexity of Exact Evaluation}
\label{subsec:cas_bottleneck}

Exact evaluation avoids numerical instability but introduces a different bottleneck: expression swell. At roots of unity, quantum integers can be represented exactly in cyclotomic fields $\Q(\zeta_{2h})$ \cite{habiro2004cyclotomic}. However, in standard computer algebra systems (CAS), this requires representing factorials as polynomials in $q$ and performing repeated rational operations \cite{Abramov1993}.

Each summation step combines rational functions as
\begin{equation}
\frac{P_A(q)}{P_B(q)} + \frac{P_C(q)}{P_D(q)} = \frac{P_A(q)P_D(q) + P_C(q)P_B(q)}{P_B(q)P_D(q)},
\end{equation}
leading to rapid growth in both degree and coefficient size. Simplification requires polynomial greatest common divisor (GCD) computations, which are themselves costly. As a result, intermediate expressions become significantly larger than the final result. This phenomenom, known as expression swell, is well documented in symbolic computation \cite{vonZurGathen2013, Geddes1992,burton2018algorithms}. It is inherent to the polynomial representation and leads to rapid growth in memory usage and runtime.

Symbolic summation algorithms can reduce or transform $q$-hypergeometric expressions \cite{zeilberger1990holonomic}, however, they do not alter the evaluation complexity of individual terms once expanded, nor do they mitigate cancellation effects in finite-precision arithmetic.

\subsection{Origin of Computational Complexity}
\label{subsec:representation_problem}

The preceding analysis reveals that both numerical instability and symbolic expression swell arise from evaluating the $q$-hypergeometric series in an expanded polynomial representation or rational form. In this representation, multiplicative structure is destroyed by expansion into dense polynomials in $q$. Moreover, common factors are introduced but only eliminated after expensive simplification and cancellation is deferred to the final stages of computation.

These effects are not artifacts of implementation, but consequences of representing inherently multiplicative objects in an additive algebraic form. Thus, the computational complexity of $q$-hypergeometric series is largely a consequence of representation choice rather than intrinsic combinatorial difficulty. An effective computational framework must therefore preserve multiplicative structure throughout, enable cancellation prior to expansion and avoid dense polynomial arithmetic altogether.

This observation motivates the representation introduced in the next section, in which quantum integers and factorials are encoded directly via their irreducible cyclotomic factors. In this formulation, multiplicative algebra is reduced to sparse integer arithmetic, eliminating both expression swell and representation-induced numerical instability.

\section{Cyclotomic Representation and the Deferred Evaluation}
\label{sec:cyclo_architecture}

In this section, we introduce a representation that directly resolves the structural and computational limitations identified in Section~\ref{sec:background}. The central idea is to separate the algebraic structure of $q$-hypergeometric series from their numerical realization by encoding quantum integers and factorials through their cyclotomic factorization. These factorizations are represented as sparse integer exponent vectors, ensuring that multiplicative structure is preserved throughout the computation.

This approach eliminates the need to construct dense polynomial expressions altogether. As a result, cancellations that were previously deferred to late stages of symbolic or numerical evaluation are instead performed exactly and immediately at the level of integer exponents.

This representation naturally leads to a \emph{deferred evaluation strategy} where the algebraic structure of the series is compiled once into an abstract combinatorial object, and evaluation in a specific numerical or algebraic field is postponed to a final projection step. In this way, the representation itself acts as a bridge between exact and approximate computation, while avoiding the intermediate complexity that plagues traditional approaches.

\subsection{Cyclotomic Factorization of Quantum Integers and Factorials}

We begin with the quantum integer
\begin{equation}
[n]_q = \frac{q^n - q^{-n}}{q - q^{-1}}.
\end{equation}
Factoring out the phase $q^{1-n}$ yields
\begin{equation}\label{eq:qint2}
[n]_q = q^{1-n} \frac{1 - (q^2)^n}{1 - q^2}  = q^{1-n} \sum_{k=0}^{n-1} q^{2k}.
\end{equation}
It is a classical result that the polynomial $x^n - 1$ admits a unique factorization into cyclotomic polynomials \cite{washington97},
\begin{equation}\label{eq:cyclo_factor}
x^n - 1 = \prod_{d \mid n} \Phi_d(x),
\end{equation}
where each $\Phi_d(x)$ is irreducible over $\Q$.

\begin{proposition}[Cyclotomic factorization]
\label{prop:cyclo_fact}
For any positive integer $n \geq 1$, the quantum integer and quantum factorial admit the factorization
\begin{align}\label{eq:qint_cyclo}
[n]_q &= q^{1-n} \prod_{\substack{d \mid n \\ d > 1}} \Phi_d(q^2), \\
[n]_q! &= q^{\frac{n(1-n)}{2}} \prod_{d=2}^n \Phi_d(q^2)^{\lfloor n/d \rfloor}
\label{eq:qfact_cyclo}
\end{align}
respectively, where $\Phi_d$ denotes the $d$-th cyclotomic polynomial. 
\end{proposition}

\begin{proof}
Equation~\eqref{eq:qint_cyclo} follows by substituting $x = q^2$ into~\eqref{eq:cyclo_factor} and cancelling $\Phi_1(q^2) = q^2 - 1$
against the denominator of~\eqref{eq:qint2}.

For the factorial, apply~\eqref{eq:qint_cyclo} to each factor of $[n]_q! = \prod_{m=1}^n [m]_q$. The phase exponent collects into the sum $\sum_{m=1}^n (1-m) = n - \tfrac{n(n+1)}{2} = \tfrac{n(1-n)}{2}$. The exponent of $\Phi_d(q^2)$ counts the number of integers $m \in \{1,\ldots,n\}$ with $d \mid m$, which is exactly $\lfloor n/d \rfloor$.
\end{proof}

The key consequence of Proposition~\ref{prop:cyclo_fact} is that quantum integers and factorials admit a \emph{canonical multiplicative decomposition} into irreducible cyclotomic components. Unlike polynomial expansions, this factorization is sparse, structured, and directly aligned with the arithmetic of roots of unity. It is precisely this structure that is obscured in conventional representations and recovered here.

\subsection{Cyclotomic Exponent Basis}

The factorizations above suggest representing algebraic expressions through their exponents with respect to the basis $\{q, \Phi_d(q^2)\}$. We formalize this as follows: 

\begin{definition}[Cyclotomic exponent vector]
Let $D_{\max} \in \mathbb{N}$. A \textbf{cyclotomic exponent vector} is a map
\begin{equation}
\mathbf{e} : \{2,3,\dots,D_{\max}\} \to \Z.
\end{equation}
Equivalently, fixing the natural ordering of indices, we identify $\mathbf{e}$ with the vector $(e_2, e_3, \dots, e_{D_{\max}}) \in \mathbb{Z}^{D_{\max}-1}$,
with sparse support. The support size is defined by
\[ \|\mathbf{e}\|_0 = \#\{d : e_d \neq 0\}. \]
\end{definition}
The exponent vectors are stored in sparse form (e.g.\ hash maps or compressed index-value arrays), so that arithmetic operations scale linearly in the support size $\|\mathbf{e}\|_0$. This replaces dense polynomial algebra with sparse integer arithmetic, with direct consequences for both computational complexity and numerical stability.

\begin{definition}[Cyclotomic monomial]
Let $q$ be a formal variable. A \textbf{cyclotomic monomial} $\mathcal{M}$ is a tuple
\begin{equation}
\mathcal{M} = (\sigma, P, \mathbf{e}),
\end{equation}
where $\sigma \in \{-1,1\}$, $P \in \Z$, and $\mathbf{e}$ is a cyclotomic exponent vector.\footnote{For generalized or weighted cyclotomic monomials, $\sigma$ may be extended to $\Q$ or $\R$, while retaining the same exponent structure.} It represents the formal expression
\begin{equation}
\label{eq:cyclo_monomial}
\mathcal{M}(q) = \sigma \, q^P \prod_{d \ge 2} \Phi_d(q^2)^{e_d}.
\end{equation}
\end{definition}

The cyclotomic monomials thus, form a free abelian multiplicative group on formal generators $q$ and $\{\Phi_d\}$, with integer exponent vectors. This representation converts multiplicative algebra into additive integer structure. In particular, multiplication and division correspond to component-wise integer operations
\begin{align}\label{eq:integer_ops}
\mathcal{M}_A \times \mathcal{M}_B &\mapsto (\sigma_A \sigma_B, \; P_A + P_B, \; \mathbf{e}_A + \mathbf{e}_B), \nonumber \\
\frac{\mathcal{M}_A}{\mathcal{M}_B} &\mapsto (\sigma_A \sigma_B, \; P_A - P_B, \; \mathbf{e}_A - \mathbf{e}_B).
\end{align}
Crucially, cancellations that would require polynomial GCD computations in the standard representation are realized here as exact integer subtractions. This directly addresses the representation-level inefficiencies identified in Section~\ref{subsec:representation_problem}.

Since the exponent vector $\mathbf{e}$ inherently allows negative integers (arising  from denominator factors in the ratio ${\cal R}_z$) the cyclotomic monomials do not reside in a standard polynomial ring. Rather, their natural algebraic setting is a localization. By formally adjoining the inverses of the deformation parameter and the cyclotomic polynomials to the polynomial ring, we obtain the localized ring $\mathbb{Z}[q^{\pm 1},\,\Phi_d(q^2)^{\pm 1}]_{d \geq 2}$. This natively accommodates the exact fractional arithmetic of the DCR without requiring expansion into dense rational functions.

\subsection{Hypergeometric Update Structure}

To avoid recomputation of large factorial expressions, the $q$-hypergeometric series Eq.~\eqref{eq:generic_q_series} are evaluated using their defining recurrence relations between consecutive summands. Let $\mathcal{M}_z$ denote the cyclotomic representation of the $z$-th summand. Then
\begin{equation}
\mathcal{M}_{z+1} = \mathcal{M}_z \cdot \mathcal{R}_z,
\end{equation}
where $\mathcal{R}_z$ is a rational function involving only local quantum integers.
For the quantum $6j$-symbol, one obtains the explicit ratio
\begin{equation}\label{eq:ratio_Rz}
\mathcal{R}_z = - \frac{[z+2]_q \prod_{y=1}^3 [b_y- z]_q}
{\prod_{i=1}^4 [z+1 - a_i]_q}.
\end{equation}
This recurrence structure is essential: it ensures that the global factorial growth observed in Section~\ref{subsec:cancellation} is never explicitly realized. Instead, the computation proceeds through local multiplicative updates.

In the cyclotomic representation, each factor in $\mathcal{R}_z$ can be decomposed independently, and each update reduces to sparse integer addition
\begin{equation}
\mathbf{e}_{z+1} = \mathbf{e}_z + \mathbf{e}_{\mathcal{R}_z}, \q \q P_{z+1} = P_z + P_{\mathcal{R}_z}
\end{equation}
with intial data at $z=z_{\min}$. The cost of propagating the full summation is therefore governed by the number of nonzero entries in these exponent vectors and the length of the summation range.

\subsection{Deferred Cyclotomic Representation}

The construction above naturally lead to a separation between algebraic compilation and numerical evaluation. For a given finite $q$-hypergeometric series, we encode the algebraic structure of each summand in a cyclotomic exponent basis and organize the result into a compact representation: 
\begin{definition}[Deferred cyclotomic representation (DCR)]
Let $S(q)$ be a finite $q$-hypergeometric series with summation range $z \in [z_{\min}, z_{\max}]$. A \textbf{deferred cyclotomic representation} of $S$ is the tuple
\begin{equation}
\mathcal{S}_{\rm DCR} = \left( \mathcal{M}_{\text{base}}, \{ {\mathcal{R}_z}\}_{z=z_{\min}}^{z_{\max}-1}, \mathcal{M}_{\text{root}}, \mathcal{M}_{\text{rad}}
\right),
\end{equation}
where $\mathcal{M}_{\mathrm{base}}$ is the cyclotomic monomial of the initial summand at $z_{\min}$, $\{\mathcal{R}_z\}$ is the sequence of multiplicative update ratios computed via Eq.~\eqref{eq:ratio_Rz}, and $(\mathcal{M}_{\mathrm{root}},\, \mathcal{M}_{\mathrm{rad}})$ are the square and square-free parts of the global geometric prefactor (as in Eq.~\eqref{eq:sqrt_decomp}). This object encodes the full algebraic structure of the series independently of any numerical value of $q$.
\end{definition}

The DCR transforms the evaluation problem of $q$-hypergeometric series $S(q)$ into two distinct stages (See Figure\ref{fig:dcr_architecture}):
\begin{itemize}
\item \emph{Compilation stage:} construction of a combinatorial object $\mathcal{S}_{\text{DCR}}$ over $\Z$ encoding all multiplicative structure,
\item \emph{Projection stage:} evaluation in a chosen target field.
\end{itemize}
This separation directly resolves the dual bottlenecks identified in Section~\ref{sec:background}, i.e., expression swell is avoided because no dense expressions are formed, and numerical instability is mitigated because cancellations are performed exactly prior to evaluation.

We use the term deferred cyclotomic representation (DCR) in two closely related senses, as the underlying combinatorial exponent data, and as the two-stage framework in which this data is constructed and subsequently evaluated via a projection map. The intended meaning will be clear from context.

\medskip
\paragraph*{Exact square-root decomposition.}
A distinctive feature of quantum recoupling coefficients like Eq.~\eqref{eq:racah_6j} is the presence of square roots of quantum factorial ratios in global geometric prefactors, such as the triangle coefficients $\Delta_{abc}$ of Eq.~\eqref{eq:tri_coeff}. In the polynomial representation, such square roots require algebraic field extensions. One cannot remain in $\mathbb{Q}(q)$ and must adjoin square roots of polynomials, inflating the arithmetic cost. 

In the cyclotomic monomial representation, square roots are handled exactly at the exponent level, avoiding algebraic field extensions during compilation. This is achieved by separating the exponents of the geometric prefactors denoted by $\mathcal{M}_{\rm geom}$ into even and odd parts
\begin{equation}\label{eq:sqrt_decomp}
\mathcal{M}_{\rm geom} = \mathcal{M}_{\rm root} \cdot \sqrt{\mathcal{M}_{\rm rad}}.
\end{equation}

Here, $\mathcal{M}_{\rm root}$ collects all cyclotomic factors with even exponents (a perfect square) and $\mathcal{M}_{\rm rad}$ collects the square-free part ($e_d \in \{-1,0,1\}$ for all $d$). Both components have integer exponents and are evaluated exactly by integer arithmetic. The entire square root is resolved at the exponent level, without field extension and without any polynomial computation. This is an exact, algebraic capability that has no direct analogue in the polynomial representation.

\medskip
\paragraph*{The DCR as algebraic preconditioner.}
Because all algebraic simplifications including phase factors, factorial cancellation, and square-root extraction, are performed at the integer level during compilation, before any value of $q$ is assigned, the DCR delivers the summand to the evaluation stage in its maximally reduced form. 

This has a precise computational consequence: the dynamic range of intermediate quantities is drastically reduced relative to the polynomial representation. In particular, the large intermediate cancellations to resolve the condition number $\kappa$ in Section~\ref{subsec:cancellation} are no longer realized numerically, but are instead resolved exactly at the algebraic level.
In this sense, the DCR acts as an \emph{algebraic preconditioner}. It reorganizes the computation so that the representation itself performs the cancellations that would otherwise occur approximately under finite precision arithmetic. This explains the qualitative improvement in numerical stability and the extension of reliable double-precision computation demonstrated in Section~\ref{sec:numerical_results}.

\begin{figure*}[htpb]
\centering
\begin{tikzpicture}[
node distance=2.5cm and 3cm,
box/.style={draw, rectangle, rounded corners, minimum height=1.2cm, align=center, thick, fill=white, inner sep=1ex},
arrow/.style={->, thick, >=stealth}
]

\node[box, fill=blue!10] (input) {$q$-hypergeometric series \\ $S(q)$};

\node[box, fill=purple!15, right=of input] (dcr) {Deferred Cyclotomic \\ Representation \\  $\mathcal{S}_{\text{DCR}}$ };

\node[box, fill=orange!10, right=of dcr, yshift=1.8cm] (exact) {Exact Field \\ $\mathbb{K} = \Q(\zeta_{2h})$};
\node[box, fill=orange!10, right=of dcr] (numeric) {Floating-point \\ $\mathbb{K} \in \{\R, \C\}$};
\node[box, fill=orange!10, right=of dcr, yshift=-1.8cm] (classical) {Classical limit \\ ($q \to 1$)};

\draw[arrow] (input) -- node[above, align=center, font=\small, text=black!80] { Cyclotomic} node[below, align=center, font=\small, text=black!80] {factorization} (dcr);

\draw[arrow, dashed, black!70] (dcr.east) to[out=0, in=180] node[above, sloped, font=\small] {Projection $\Pi_q$} (exact.west);
\draw[arrow, dashed, black!70] (dcr.east) -- node[above, font=\small] {} (numeric.west);
\draw[arrow, dashed, black!70] (dcr.east) to[out=0, in=180] node[below, sloped, font=\small] {} (classical.west);

\end{tikzpicture}
\caption{The deferred cyclotomic architecture. A $q$-hypergeometric series is compiled once into the DCR as a parameter-independent combinatorial object over $\Z$ encoding its full algebraic content without assuming any value of $q$. Evaluation in a target field $\mathbb{K}$ is a projection $\Pi_q$ applied to this fixed object. Exact arithmetic, floating-point evaluation, and the classical limit $q \to 1$ are all distinct instances of this projection map.}
\label{fig:dcr_architecture}
\end{figure*}
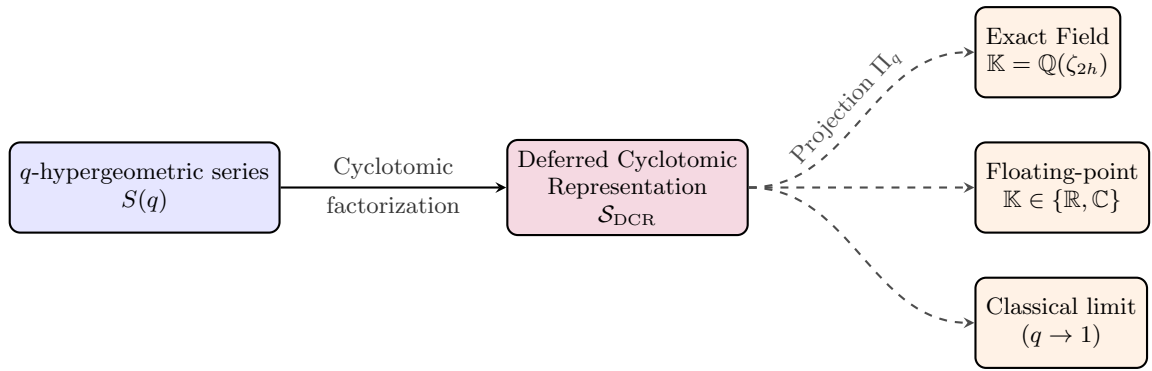

We summarize the construction of the deferred representation for a generic $q$-hypergeometric series in Algorithm \ref{alg:compilation}.

\begin{algorithm}[H]
\caption{Deferred cyclotomic compilation}
\label{alg:compilation}
\begin{algorithmic}[1]
\Require A $q$-hypergeometric series $S(q)$ with external parameters (e.g. spins $j$)
\Ensure Deferred cyclotomic representation  $\mathcal{S}_{\text{DCR}}$
\State $z_{\min}, z_{\max} \gets$ determine summation bounds from admissibility conditions
\State $\mathcal{M}_{\text{geom}} \gets$ construct geometric prefactors (e.g.\ $\Delta_{abc}$) via Proposition~\ref{prop:cyclo_fact}
    \Comment{integer arithmetic only}   
\State $(\mathcal{M}_{\text{root}}, \mathcal{M}_{\text{rad}}) \gets$ extract perfect squares from $\mathcal{M}_{\text{geom}}$
\State construct $\mathcal{M}_{\mathrm{base}}$ at $z = z_{\min}$
\State $\{\mathcal{R}_z\} \gets []$ \Comment Initialize ratio sequence  
\For{$z = z_{\min}$ to $z_{\max}-1$}
    \State $\mathcal{R}_z \gets$ compute cyclotomic factorization of the ratio via Eq.~\eqref{eq:ratio_Rz}
    \State Append $\mathcal{R}_z$ to $\{\mathcal{R}_z\}$
\EndFor
\State \Return $(\mathcal{M}_{\mathrm{base}}, \{\mathcal{R}_z\}, \mathcal{M}_{\text{root}}, \mathcal{M}_{\text{rad}})$
\end{algorithmic}
\end{algorithm}

Once constructed, the DCR can be specialized to different target domains through a projection step. The same representation is reused in all cases, while only the evaluation map depends on the chosen numerical or algebraic setting. This separation allows a single compiled object to support consistent evaluation across exact and approximate regimes, and underlies the stability and efficiency gains analyzed below. In this sense, the DCR acts as a representation-level preconditioner, reorganizing the computation prior to any numerical or symbolic realization.

\section{Universal Field Projections}
\label{sec:universal_projections}

The deferred cyclotomic representation (DCR) provides a unified representation of $q$-hypergeometric series in terms of cyclotomic exponent data. Evaluation is realized as a projection of the representation into a chosen target field $\mathbb{K}$ by assigning concrete values to the cyclotomic basis elements.

This formulation elevates evaluation from a computational procedure to a structural operation. Rather than constructing distinct algebraic expressions for each specialization of $q$, the DCR defines a single universal object together with a family of projection maps. In this sense, $q$-dependence is no longer an intrinsic feature of the representation itself, but an external parameter introduced only at the level of evaluation. This separation is the key conceptual mechanism underlying both the computational efficiency and the structural consequences developed later in Section \ref{sec:structure}.

More precisely, each cyclotomic monomial $\mathcal{M}$ Eq.~\eqref{eq:cyclo_monomial} corresponding to a summand in the series, is encoded by a sign $\sigma$, an integer power $P$ of the deformation parameter, and a sparse collection of exponents $\{e_d\}$. A projection into $\mathbb{K}$ is specified by the following definition.
\begin{definition}[Field projection]
Let $\mathbb{K}$ be a target field. A \emph{field projection} evaluated at a non-zero parameter $q \in \mathbb{K}^\times$ is a map
\[ \Pi_{q} : \mathcal{M} \;\to\; \mathbb{K} \]
defined on cyclotomic monomials by
\begin{equation} \label{eq:project_monomial}
\Pi_{q}(\mathcal{M}) = \sigma\, q^P \prod_{d \ge 2} \Phi_d(q^2)^{e_d},
\end{equation}
where $\sigma \in \{\pm 1\}$, $P \in \Z$, and $\{e_d\}$ are the cyclotomic exponents. The images of $q$ and $\Phi_d(q^2)$ are dictated by the algebraic structure of the target field $\mathbb{K}$.
\end{definition}

The map $\Pi_{q}$ is multiplicative by construction and extends linearly to sums of cyclotomic monomials, ensuring compatibility with the hypergeometric summation structure. Thus, the evaluation of a full $q$-hypergeometric series proceeds by iterating the update relations encoded in the DCR and applying $\Pi_{q}$ to each intermediate monomial. Since the representation itself is independent of $\mathbb{K}$, the same compiled object can be reused across all projection regimes. This perspective recasts evaluation as the application of a family of parameter-dependent maps to a fixed combinatorial object.

\subsection{Projection Regimes}

We distinguish several natural projection regimes relevant to applications in quantum topology and numerical analysis.

\paragraph*{1. Root-of-unity evaluation.}
Let $q = \exp(i\pi/h)$, so that $q^2 = \exp(2\pi i / h)$ is a primitive $h$-th root of unity. Physically, evaluating quantum recoupling coefficients at a root of unity corresponds to introducing a non-zero cosmological constant $\Lambda \propto 1/h$ into the gravitational path integral, yielding a semiclassical limit governed by curved geometries \cite{pranzetti2014turaev,Haggard:2014xoa,Haggard:2015yda,bonzom2014deformed}. Such $q$-deformations are also deeply tied to the invariance of topological state sums under coarse-graining \cite{Livine:2016vhl,Dittrich:2011zh,Asante:2022dnj,Livine:2013gna}. 

In this regime, the cyclotomic polynomials satisfy
\begin{equation}
\Phi_d(q^2) = 0 \quad \text{if and only if } d = h.
\end{equation}
Thus, the vanishing locus is supported on a single cyclotomic index. Any cyclotomic monomial with an exponent $e_h > 0$ evaluates to zero. This leads to immediate simplifications, since contributions can be identified and eliminated early, allowing early termination in the evaluation of individual summands.  Conversely, monomials with $e_h < 0$ would manifest as poles under this projection. However, such poles do not appear in physically admissible amplitudes. The geometric triangle inequalities and quantum admissibility constraints inherently guarantee that all cyclotomic denominators in the update sequence $\{\mathcal{R}_z\}$ are fully canceled by corresponding numerators at the singular index $h$. Therefore, the projection map $\Pi_{\mathbb{K}}$ remains well-defined on the admissible subring generated by the DCR.

Thus, despite the presence of arbitrarily large indices $d$ in the deferred cyclotomic representation, only the single index $d = h$ governs singular behaviour upon specialization. This sharply reduces the number of nonzero terms and is particularly effective in the computation of Turaev--Viro type invariants and exact finite-group spin foams.

\paragraph*{2. Exact algebraic projection.}
For symbolic evaluation, one may take $\mathbb{K} = \Q(\zeta_{2h})$, the cyclotomic number field generated by a primitive root $\zeta_{2h}$. The projection
\begin{equation}
\Pi_{\zeta}(\mathcal{M}) = \sigma\, \zeta^{P} \prod_{d \ge 2} \Phi_d(\zeta^2)^{e_d}
\end{equation}
is then computed using exact arithmetic. Because cancellations between numerator and denominator are resolved at the level of exponent vectors, the projection avoids intermediate polynomial division and reduces reliance on expensive greatest common divisor computations. As a result, it is well suited for implementations in computer algebra systems such as \texttt{Nemo.jl} \cite{fieker2017nemo}.

\paragraph*{3. Complex analytic evaluation.}
For generic $q \in \mathbb{C}^\times$, redundant recomputation of cyclotomic factors is avoided by precomputing the values $\{\Phi_d(q^2)\}$ up to the maximal degree required by the deferred representation. Using the identity
$ q^{2n} - 1 = \prod_{d \mid n} \Phi_d(q^2),$ the cyclotomic basis is generated via recursive Möbius inversion,
\begin{equation}
\Phi_n(q^2) = \frac{q^{2n} - 1}{\prod_{d \mid n,\, d < n} \Phi_d(q^2)}.
\end{equation}
This produces a dense table in $\mathcal{O}(D \log D)$ time, eliminating repeated factorization or symbolic expansion.

Once precomputed, any deferred cyclotomic monomial $\mathcal{M}$ is evaluated through sparse multiplicative assembly according to Eq.~\eqref{eq:project_monomial}. The computation reduces to fast exponentiation and a small number of complex multiplications, yielding an efficient and numerically stable evaluation pipeline, particularly near the unit circle where direct methods are prone to loss of precision.

\paragraph*{4. Classical limit.}
In the limit $q \to 1$, the cyclotomic factors satisfy \cite{washington97}
\begin{equation}\label{eq:classical_limit}
\Phi_d(1) = \begin{cases}
p & \text{if } d = p^m \text{ for some prime } p, \\
1 & \text{otherwise.}
\end{cases}
\end{equation}
Under this specialization (the projection $\Pi_1$), the cyclotomic exponent representation reduces to a prime factorization of classical factorial expressions. This establishes consistency between the deferred cyclotomic framework and standard arithmetic formulations of Racah--Wigner coefficients \cite{Johansson2016}.

\begin{algorithm}[H]
\caption{DCR Projection}
\label{alg:projection}
\begin{algorithmic}[1]
\Require $\mathcal{S}_{\mathrm{DCR}}$, evaluation parameter $q$
\Ensure  Evaluated amplitude $S(q) \in \mathbb{K}$
\State Precompute $\Pi_{q}(\Phi_d(q^2))$ for all required $d$
       \Comment{parameter-dependent, one-time cost}
\State $S_{\rm sum} \gets \Pi_{q}(\mathcal{M}_{\rm base})$
       \Comment{$z_{\min}$ term}
\State $R \gets 1$
\For{each ratio $\mathcal{R}_z$ in $\{\mathcal{R}_z\}$}
  \State $v \gets \Pi_{q}(\mathcal{R}_z)$
  \If{$v = 0$}
    \textbf{break}
    \Comment{early termination at vanishing factor}
  \EndIf
  \State $R \gets R \cdot v$
  \State $S_{\rm sum} \gets S_{\rm sum} + \Pi_{q}(\mathcal{M}_{\rm base}) \cdot R$
\EndFor
\State \Return
  $\sqrt{\Pi_{q}(\mathcal{M}_{\rm rad})} \cdot \Pi_{q}(\mathcal{M}_{\rm root}) \cdot S_{\rm sum}$
\end{algorithmic}
\end{algorithm}

These projection regimes illustrate that numerical evaluation, exact arithmetic, and classical limits are not distinct computational problems, but instances of a single universal projection mechanism acting on a fixed combinatorial object.

An important advantage of the DCR is the square-free decomposition $\mathcal{M}_{\mathrm{geom}} = \mathcal{M}_{\mathrm{root}} \cdot \sqrt{\mathcal{M}_{\mathrm{rad}}}$ performed during compilation. While $\mathcal{M}_{\mathrm{root}}$ is evaluated using standard integer arithmetic, the residual factor $\sqrt{\mathcal{M}_{\mathrm{rad}}}$ requires a branch choice during projection. We evaluate this by assigning to each remaining square-free factor $\Phi_d(q^2)^{e_d}$ its projection $\Pi_{\mathbb{T}}(\Phi_d(q^2))^{e_d}$, and subsequently taking the principal square root in the target field $\mathbb{T}$. For exact algebraic projections where $\mathbb{T} = \mathbb{Q}(\zeta_{2h})$, we select the unique positive square root of the rational component paired with the standard branch for cyclotomic units, consistent with classical Racah--Wigner recoupling conventions \cite{varshalovich1988quantum}. 

\section{Structural and Algebraic Consequences of the DCR}
\label{sec:structure}

The deferred cyclotomic representation (DCR) has been introduced as a computational framework for stabilizing and scaling the evaluation of $q$-hypergeometric series. However, by isolating algebraic structure from evaluation, it also provides a new perspective on the internal organization of quantum amplitudes. In particular, deformation, admissibility, and asymptotic behavior can be reformulated as properties of a fixed combinatorial object prior to evaluation.

We emphasize here that several of the structural interpretations developed in this section are preliminary. While the algebraic statements regarding the DCR are exact, their implications for quantum topology and discrete quantum geometry require further investigation. The goal here is therefore to identify precise reformulations and to delineate directions where the representation may offer new conceptual leverage.

\subsection{Deformation as Evaluation of a Fixed Combinatorial Object}

A defining feature of the DCR is that all dependence on the deformation parameter $q$ is mediated exclusively through the projection map. The underlying exponent data is independent of $q$.

\begin{proposition}[Deformation as evaluation]
\label{prop:deformation_as_evaluation}
Let $S(q)$ be a finite $q$-hypergeometric series admitting a deferred cyclotomic representation $\mathcal{S}_{\mathrm{DCR}}$. Then for any admissible value of $q$,
\begin{equation}
S(q) = \Pi_q(\mathcal{S}_{\mathrm{DCR}}),
\end{equation}
where $\Pi_q$ denotes the evaluation map defined in Eq.~\eqref{eq:project_monomial}.
\end{proposition}

\begin{proof}
By construction, each summand of $S(q)$ is represented as a cyclotomic monomial $\mathcal{M}_z$ whose evaluation under $\Pi_q$ reproduces the original expression. The full sum is obtained by iterating the update relations encoded in $\mathcal{S}_{\mathrm{DCR}}$ and applying $\Pi_q$ termwise. Since $\Pi_q$ is multiplicative, the result follows.
\end{proof}

This formulation implies that quantum amplitudes at different values of $q$ are not distinct algebraic objects, but rather evaluations of a single underlying combinatorial structure. In particular, all topological levels $k$ and the classical limit $q \to 1$ arise from the same data $\mathcal{S}_{\mathrm{DCR}}$ via different projection maps.

In quantum topology and discrete quantum gravity models, the deformation parameter $q$ is related to physical parameters such as the cosmological constant $\Lambda$. For example, in Turaev--Viro type models, choosing $q$ to be a root of unity corresponds to a theory with positive cosmological constant ($\Lambda > 0$), while other regimes are associated with different geometric behaviors.

From the perspective of the DCR, these different physical regimes arise from different evaluations of the same combinatorial data. The representation itself does not encode the value or sign of $\Lambda$. Instead, this information enters only through the projection map. It is important to emphasize that this does not imply that different physical theories are equivalent. Rather, the DCR shows that they share a common algebraic backbone, while their physical interpretation is determined by how this structure is evaluated.

\subsection{Admissibility as Cyclotomic Vanishing}

At roots of unity, admissibility conditions restrict the allowed spins and summation ranges. In the DCR, these constraints acquire an intrinsic algebraic characterization.

\begin{proposition}[Admissibility via cyclotomic support]
\label{prop:admissibility}
Let $q = e^{i\pi/h}$, so that $q^2$ is a primitive $h$-th root of unity. For a cyclotomic monomial $\mathcal{M}$ with exponent vector $\mathbf{e}$, the evaluation satisfies
\begin{equation}
\Pi_q(\mathcal{M}) = 0 \q \Longleftrightarrow \q e_h > 0.
\end{equation}
\end{proposition}

\begin{proof}
At $q^2 = e^{2\pi i/h}$, one has $\Phi_d(q^2) = 0$ if and only if $d = h$. Thus $\mathcal{M}(q)$ vanishes precisely when the exponent of $\Phi_h(q^2)$ is positive.
\end{proof}

This identifies admissibility as a vanishing condition internal to the representation, rather than an externally imposed constraint. Terms violating admissibility are annihilated under projection, and in practical implementations this can lead to early termination of summations. Conceptually, this reframes admissibility as a property of cyclotomic support. The role of the root of unity is to probe this support via the projection map, rather than to impose constraints at the level of summation indices.

\subsection{Algebraic Structure and Exponent Geometry}

Within the DCR, each term in the hypergeometric summation is encoded by an exponent vector $\mathbf{e}_z$ together with a phase variable $P_z$. The summation therefore defines a discrete trajectory
\begin{equation}
z \;\longmapsto\; (\mathbf{e}_z, P_z) \in \Z^{D_{\max}} \times \Z.
\end{equation}
whose evolution is governed by the rational update operator $\mathcal{R}_z$.

This perspective reformulates multiplicative algebra of quantum factorials as additive arithmetic on integer lattices. In particular,
\begin{itemize}
\item cancellations correspond to linear relations in $\Z^{D_{\max}}$,
\item factorial growth of expressions is encoded in the norm $\|\mathbf{e}_z\|$,
\item phase information is tracked separately through $P_z$.
\end{itemize}

This suggests an interpretation of $q$-hypergeometric summation as a discrete flow in exponent space. The structure of this flow is independent of the choice of target field and may provide a useful organizing principle for combinatorial identities and analyzing asymptotic behaviour, particularly in regimes where traditional analytic methods become unwieldy.

\subsection{Coherence Identities at the Representation Level}

Quantum recoupling coefficients satisfy several coherence relations. For instance orthogonality conditions (bubble move) and the Biedenharn-Elliott pentagon identity, which ensures the triangulation independence of topological invariants \cite{varshalovich1988quantum,Biedenharn1953,pachner1991pl}. Conventionally, these identities are established as analytical equalities between evaluated physical amplitudes.

Because the DCR isolates amplitudes into canonical combinatorial objects $\mathcal{S}_{\mathrm{DCR}}$, it raises the possibility that such identities may admit a formulation directly at the representation level. In particular, one may ask whether the pentagon identity can be realized as an equality between unprojected cyclotomic exponent data. At present, this remains an open question. A positive answer would imply that topological invariance is rooted in purely combinatorial properties of the exponent structure, independent of analytic features of the deformation parameter. The DCR provides a concrete setting in which this question can be formulated precisely.

\subsection{Semiclassical Limit and Exponent Asymptotics}

The semiclassical behaviour of $q$-deformed amplitudes is typically analyzed via saddle-point analysis of factorial or polynomial expressions \cite{Ponzano1968,Roberts1999,Taylor2005}. In the DCR, factorials are replaced by cyclotomic exponent data, providing an alternative approach for investigating this limit. 

For a cyclotomic monomial $\mathcal{M}$, its logarithm under projection can be into magnitude and phase components as
\begin{equation}
\log \Pi_q(\mathcal{M}) = P \log q + \sum_{d \ge 2} e_d \log \Phi_d(q^2).
\end{equation}
Thus, the asymptotic behaviour is governed by the growth of the exponent vector $\mathbf{e}$ coupled with the analytic properties of $\log \Phi_d(q^2)$.

In the classical limit $q \to 1$, the cyclotomic factors reduce to prime contributions according to Eq.~\eqref{eq:classical_limit}, and the exponent data recovers the prime factorization of classical factorials via Legendre's formula
\begin{equation}
\nu_p(n!) = \sum_{m=1}^{\infty} \left\lfloor \frac{n}{p^m} \right\rfloor, \q \text{for $p$ prime}.    
\end{equation} 
This establishes a direct interpolation between quantum and classical regimes at the level of exponent structure. 

Understanding how the discrete trajectory of the exponent space $(\mathbf{e}_z, P_z)$ encodes macroscopic geometric quantities represents an interesting direction for discrete quantum gravity.In the classical limit, the Ponzano--Regge asymptotic formula establishes that the amplitude oscillates as $\sim (12\pi V)^{-1/2} \cos(S_R + \pi/4)$, where $V$ is the volume of the Euclidean tetrahedron and $S_R$ is the discrete Regge action \cite{Ponzano1968,Roberts1999}. Within the DCR framework, the large-$j$ growth of the exponent vectors, governed by Legendre's formula, naturally controls the scaling of the compiled factorial data and therefore provides the combinatorial structure from which the semiclassical prefactor may ultimately be recovered.

The emergence of the oscillatory phase $S_R$ is traditionally accessed via stationary phase approximations. The DCR suggests a purely combinatorial alternative where the oscillatory Regge action must be analytically encoded within the interference pattern of the phase trajectory $P_z$ and the alternating signs distributed over the sparse cyclotomic support. Furthermore, for $q$ evaluated at a root of unity, the Taylor--Woodward asymptotic formula \cite{Taylor2005} plays the analogous role, where curved simplex geometry replaces Euclidean geometry. Because the DCR completely isolates the $q$-dependence into the projection map, it implies that these curved-geometry asymptotics emerge directly from the analytic properties of $\log \Phi_d(q^2)$ evaluated over the same parameter-independent integer exponents $(\mathbf{e}_z, P_z)$.
Formally extracting this geometric action directly from the discrete integer exponent flow, without relying on continuous analytical approximations, stands as a well-posed open problem. The DCR isolates the relevant combinatorial data in a form perfectly amenable to this analysis, establishing an algebraic foundation for future investigations into the continuum limit of discrete quantum gravity.

\subsection{Connections to Cyclotomic Structures and Quantum Topology}
\label{subsec:cyclotomic_connections}

Cyclotomic structures play a central role in the arithmetic of quantum topology. In particular, quantum invariants of $3$-manifolds evaluated at roots of unity satisfy remarkable integrality properties and admit unified formulations within cyclotomic completions such as Habiro's ring $\widehat{\mathbb{Z}[q]}$ \cite{habiro2004cyclotomic,Gilmer_tqft,garoufalidis2005colored}. Habiro's construction provides a universal algebraic framework in which quantum invariants are represented by a single object whose evaluation at every root of unity yields an algebraic integer. This reveals a deep interplay between topology, number theory, and quantum algebra.

The DCR offers a complementary perspective that operates at a different algebraic level. Whereas Habiro's ring describes the global quantum invariant obtained after summation, the DCR resolves the arithmetic structure locally at the level of the individual $q$-hypergeometric summands. Each summand generally contains genuine denominators arising from quantum factorials and is therefore naturally represented in the localized cyclotomic ring $\mathbb{Z}\!\left[q^{\pm1},\,\Phi_d(q^2)^{\pm1}\right]_{d\ge2}$, where negative exponents record the inverse cyclotomic factors explicitly. The localization is therefore the natural algebraic setting for the intermediate terms of the computation.

This distinction clarifies the respective roles of the two frameworks. Habiro's completion characterizes the arithmetic properties of the final invariant, while the DCR provides a finite combinatorial representation that makes the intermediate cyclotomic factorization explicit throughout the summation. In particular, the exponent vectors record how the localized contributions of individual summands combine so that the poles present in intermediate expressions cancel in the final result. In this sense, the DCR separates the local arithmetic of the summands from the global arithmetic of the completed invariant, providing a computational framework in which these cancellations can be analyzed explicitly. 

The DCR also naturally separates combinatorial structure from arithmetic evaluation. The compiled exponent data are independent of the deformation parameter and of the field in which evaluation takes place.
This viewpoint also highlights the role of cyclotomic fields and their Galois structure.  Evaluation at roots of unity takes place in cyclotomic fields $\mathbb{Q}(\zeta_n)$, whose arithmetic is governed by the Galois group $\mathrm{Gal}(\mathbb{Q}(\zeta_n)/\mathbb{Q})$ \cite{washington97}. Since the exponent vectors are independent of the target field, arithmetic symmetries act only through the projection map rather than the combinatorial representation itself. Establishing a precise relationship between the DCR and cyclotomic completions such as Habiro's ring, and understanding how arithmetic symmetries are reflected through the projection process, remain interesting directions for future work.

\section{Complexity and Stability Analysis}
\label{sec:complexity_error}

In this section, we analyze the computational complexity and numerical stability of the deferred cyclotomic representation (DCR). We shall make use of the main structural feature, that is the separation between (i) the combinatorial construction of a sparse cyclotomic exponent representation and (ii) its evaluation via a projection $\Pi_q$ into a target field $\mathbb{K}$. This separation allows us to isolate intrinsic algebraic complexity from representation-dependent numerical effects.

The central claim of this section is that the DCR shifts computational complexity away from algebraic manipulation and into a controlled projection stage. In particular, it significantly reduces representation-induced complexity, both intermediate expression swell in symbolic computation and dynamic range inflation in numerical evaluation, leaving only the intrinsic complexity of the summation itself.

\subsection{Construction and Evaluation Complexity}

Let $Z = z_{\max} - z_{\min}$ denote the summation length and $D_{\max}$ the maximal cyclotomic order appearing in the factorization. Each intermediate object is represented by a sparse exponent vector $\mathbf{e} \in \Z^{D_{\max}-1}$. We assume a bit-complexity model in which multiplication of $L$-bit integers has cost $M(L)$, and floating-point arithmetic has constant cost. Each update ratio involves at most $K$ cyclotomic factors, each affecting a sparse subset of size $S_{\max}$ of the exponent vector which depends on the divisor function $\tau(n)$. 
\begin{theorem}[Construction complexity]
The compilation of the DCR combinatorial object requires
\begin{equation}
\mathcal{O}\!\left(D_{\max} + Z \cdot K \cdot S_{\max} \right) \q 
\end{equation}
integer operations and stores a representation of the same asymptotic size.
\end{theorem}

\begin{proof}[Sketch]
Precomputation of cyclotomic data up to $D_{\max}$ contributes $\mathcal{O}(D_{\max})$. Each of the $Z$ update ratios modifies at most $K$ factors, and each factor updates at most $S_{\max}$ sparse entries, yielding the stated bound.
\end{proof}
For the quantum $6j$-symbols (or generic quantum recoupling coefficients), $K$ is constant and $S_{\max}$ grows slowly with $D_{\max}$, yielding near-linear scaling in $Z$.

\medskip
Evaluation consists of projecting the exponent representation into a target field $\mathbb{K}$ via the map $\Pi_q$.
\begin{proposition}[Evaluation complexity]
Let $\mathbb{K}$ be the target field and $C_{\mathbb{K}}$ the cost of one arithmetic operation in $\mathbb{K}$. Then evaluation of the DCR requires
\begin{equation}
\mathcal{O}\!\left( D_{\max} \cdot C_{\mathbb{K}} + Z \cdot K \cdot S_{\max} \cdot C_{\mathbb{K}} \right)
\end{equation}
operations for the full summation.
\end{proposition}

For floating-point arithmetic, $C_{\mathbb{K}} = \mathcal{O}(1)$, thus evaluation is effectively $\mathcal{O}(Z)$ up to divisor-function corrections. For evaluation in exact cyclotomic number fields, elements of $\mathbb{Q}(\zeta_n)$ have dimension $\varphi(n)$, i.e., Euler's totient function, so $C_{\mathbb{K}} = \mathcal{O}(\varphi(n) \cdot M(L))$. This yields an evaluation cost $ \sim \mathcal{O}\!\left(
Z \cdot K \cdot S_{\max} \cdot \varphi(D_{\max}) \cdot M(L)
\right). $

\subsection{Numerical Stability and Conditioning}

Conventional approaches rely on \emph{eager evaluation}, which refer to a paradigm where quantum factorials are immediately expanded into dense polynomial representations or directly accumulated in floating-point arithmetic prior to any structural simplification. In exact symbolic computation, this eager expansion generates intermediate objects whose polynomial degree scales directly with $D_{\max}$. This necessitates repeated polynomial multiplication and costly GCD reductions, producing superlinear memory and runtime growth. 
In contrast, the DCR replaces these dense algebraic objects with sparse integer exponent vectors. By deferring evaluation, intermediate expression swell is eliminated, and all asymptotic computational cost is shifted to the final projection stage.

Let $ S = \sum_z (-1)^z T_z$ denote the target alternating sum. The intrinsic sensitivity of the problem is governed by the condition number $\kappa$ (Eq.~\eqref{eq:condition_num}), which is independent of representation. In floating-point arithmetic with unit roundoff $\epsilon$, the forward error satisfies
\begin{equation}\label{eq:forwd_error}
\frac{|\widehat{S} - S|}{|S|} = \mathcal{O}\!\left(\kappa \,  Z \,\epsilon\right) + \mathcal{O}\!\left(\kappa \, \gamma_{\text{rep}} \, \epsilon\right).
\end{equation}
It is important to formally distinguish these two independent sources of numerical error to properly interpret the stability benchmarks. 

The first term is governed by the condition number $\kappa$ (Eq.~\ref{eq:condition_num}), which captures the intrinsic sensitivity of the alternating sum. This quantity is an invariant property of the physical amplitude itself, thus no choice of mathematical representation can reduce it, and it establishes a strict theoretical floor on precision loss. The second term captures the representation-induced amplification. Here, $\gamma_{\mathrm{rep}}$ is defined as 
\begin{equation}\label{eq:gamma_rep}
\gamma_{\text{rep}} = \max_z (\log_{10}|N_z^{\rm rep}| + \log_{10}|D_z^{\rm rep}|),
\end{equation}
which measures the maximum dynamic range of the intermediate numerator $N_z$ and denominator $D_z$ constructed when evaluating the summand $T_z = N_z/D_z$.
Unlike $\kappa$, the amplification factor $\gamma_{\text{rep}}$ isolates the contribution of the representation choice to the total numerical error. It quantifies the artificial dynamic range introduced by the chosen algebraic form of the summands prior to any cancellation.

In eager logarithmic evaluation (LSE) based on the standard polynomial representation, each summand is computed via the floating-point subtraction
\begin{equation}\label{eq:logT}
\log T_z = \log N_z - \log D_z.
\end{equation}
Because the eager polynomial representation forces algebraic cancellations to occur only after expansion, the unreduced components $N_z$ and $D_z$ grow into exponentially large, closely matched quantities. The subtraction of these magnitudes in Eq.~\eqref{eq:logT} introduces severe rounding error, yielding large amplification factor $\gamma_{\text{eager}} \gg 1$.

In contrast, the DCR performs all cancellations symbolically at the level of exponent vectors before numerical evaluation. The quantities entering the projection step are therefore already reduced to their minimal form, avoiding subtraction of large nearly equal numbers. As a result,
\[ \gamma_{\text{DCR}} \ll \gamma_{\text{eager}}, \]
and the numerical error is driven primarily by the intrinsic condition number $\kappa$. The exact analytic relationship between these error bounds is left for future work.

\medskip

In summary, eager polynomial representations entangle combinatorial growth with evaluation, leading to both expression swell and amplified rounding error. The DCR on the other hand isolates intrinsic computational difficulty from representation-induced artifacts. It removes intermediate algebraic growth through exact symbolic cancellation and reduces numerical instability by compressing the dynamic range prior to evaluation. Consequently, the DCR achieves near-optimal scaling with respect to summation length and approaches the minimal error dictated by the intrinsic conditioning of the problem.

These improvements are not algorithmic refinements, but consequences of aligning the representation with the underlying algebraic structure. See Section~\ref{subsec:cancellation} for empirical results.

\subsection{Macroscopic Complexity and State-Sum Amplitudes}
\label{subsec:state_sums}

The preceding analysis characterizes the complexity and numerical stability of individual $q$-hypergeometric series. However, physical observables in topological quantum field theory and spin foam models arise from large-scale contractions of such amplitudes \cite{barrett2007observables,dupuis2014observables,Dittrich:2018dvs,Bonzom:2022bpv}. The principal advantage of the deferred cyclotomic representation emerges in this macroscopic regime, where it enables a separation between local combinatorial structure and global summation.

For composite amplitudes defined by tensor products, the ring homomorphism property of the projection map $\Pi_q$ ensures that
\begin{equation}
\Pi_q(S_1 \cdot S_2) = \Pi_q(S_1) \cdot \Pi_q(S_2).
\end{equation}
As a consequence, composite amplitudes factorize at the representation level. Each elementary building block (e.g.\ a quantum $6j$-symbol) can be compiled independently into its DCR form, and evaluation of composite expressions reduces to multiplication in the target field $\mathbb{K}$. This induces a \emph{compilation paradigm}, where local geometric contributions are precomputed once as reusable combinatorial objects, while global observables are obtained through parameter-dependent projection and aggregation.

\medskip
\paragraph*{Application to state-sum models.}

The computational advantages of the DCR (via amortization) become particularly relevant in settings where the same local amplitudes are evaluated repeatedly, such as topological state sums, spin-foam models, and tensor-network algorithms. The essential observation is that these frameworks separate into two distinct computational layers: a local layer involving the evaluation of recoupling coefficients and a global layer involving the summation or contraction over many local configurations. The DCR addresses the first of these layers by replacing repeated algebraic reconstruction with a compile-once, evaluate-many-times strategy.

As an example, consider the Turaev--Viro partition function on a triangulated three-manifold $\mathcal{M}$,
\begin{equation}\label{eq:TV}
\mathrm{TV}_k(\mathcal{M}) = \mathcal{N}^{-|v|} \sum_{ \{j\} }\prod_{\tau \in \Delta_3(\mathcal{M})} \sixj_q^{\tau},
\end{equation}
where the summation runs over admissible colorings $\{j\}$ and each tetrahedron contributes a quantum $6j$-symbol. $\cal N$ is a global normalization constant and $|v|$ denotes the number of vertices. In a direct implementation, every occurrence of a local amplitude requires the construction and numerical evaluation of the corresponding $q$-hypergeometric expression. For large discrete levels $k$ or large spin labels, this repeated reconstruction introduces significant algebraic overhead due to polynomial growth, expression swell, and repeated cancellation-sensitive evaluations.

The DCR reorganizes the computational workflow by separating the construction of the algebraic object from its numerical realization. For a fixed local spin configuration, the factorial structure of the corresponding amplitude is compiled once into a sparse cyclotomic exponent representation. This representation is independent of the subsequent choice of evaluation parameter or target field and can therefore be reused whenever the same local amplitude appears again. Subsequent evaluations require only the projection of the precomputed exponent data. 

This amortization is particularly valuable for tensor-network and state-sum approaches, where local recoupling coefficients form the elementary building blocks of large-scale computations. Coarse-graining and renormalization steps in spin foam models and tensor network approaches repeatedly evaluate local recoupling amplitudes across different spin configurations and levels. The DCR preserves these local amplitudes as exact algebraic objects while allowing numerical specialization to be deferred until required by the final observable.

It is important to emphasize that the DCR does not remove the intrinsic global complexity of state-sum evaluations or tensor-network contractions. The exponential growth associated with summing over internal colorings, large-scale tensor contractions, or Monte Carlo sampling remains a property of the underlying physical model. Rather, the contribution of the DCR is to eliminate the repeated local algebraic overhead that currently bottlenecks exact topological computations \cite{Jaeger1990}, large-scale spin foam simulations \cite{Engle:2007wy,Delcamp:2016yix,Dona2018,Dona:2022yyn,Asante:2024eft}, and recent efforts to extract macroscopic effective dynamics \cite{Asante:2021zzh,Asante:2020qpa,Han:2021kll,Asante:2022lnp}. In this sense, the DCR provides an algebraic preconditioning layer between the exact local amplitudes and the global combinatorial computation.

\section{Performance and Numerical Validation}
\label{sec:numerical_results}

We now empirically validate the DRC architecture by benchmarking memory usage, numerical stability, and execution latency. These results confirm the practical feasibility of the framework and validate the theoretical complexity bounds established in Section~\ref{sec:complexity_error}.

Throughout this section, we contrast the DCR against conventional \emph{eager evaluation} baselines. We use the term ``eager'' to denote any strategy that forces the algebraic or numerical resolution of the amplitude prior to structural simplification. Specifically:
\begin{itemize}
\item In exact symbolic arithmetic (Eager CAS), eager evaluation refers to the standard approach of expanding the hypergeometric series into dense polynomial or rational representations and relying on algorithmic greatest common divisor (GCD) reductions.
\item In floating-point arithmetic (Eager LSE), it refers to the direct logarithmic summation of expanded $q$-factorial expressions, where cancellations are attempted numerically rather than resolved algebraically.
\end{itemize}

All experiments use the symmetric quantum $6j$-symbol ($j_i = j$) as a representative test case (see Eq.~\eqref{eq:racah_6j}). Computations are implemented in the open-source Julia package \texttt{QRecoupling.jl} \cite{QRecoupling_pkg}\footnote{This library was developed in conjunction with this manuscript to natively implement the DCR framework and ensure exact reproducibility of all presented benchmarks.}, leveraging multiple dispatch and type stability to ensure no runtime overhead from the abstract algebraic structures. Floating-point computations use IEEE 754 double precision (\texttt{Float64}), while arbitrary precision uses \texttt{BigFloat} (MPFR). Exact algebraic evaluation uses \texttt{Nemo.jl}, interfacing with FLINT and Antic \cite{fieker2017nemo}. Benchmarks were run on an Apple Silicon M2 system with 16 GB memory.

\subsection{Memory Usage and Complexity}

\begin{figure*}[htpb]
\centering
\includegraphics[width=0.65\textwidth]{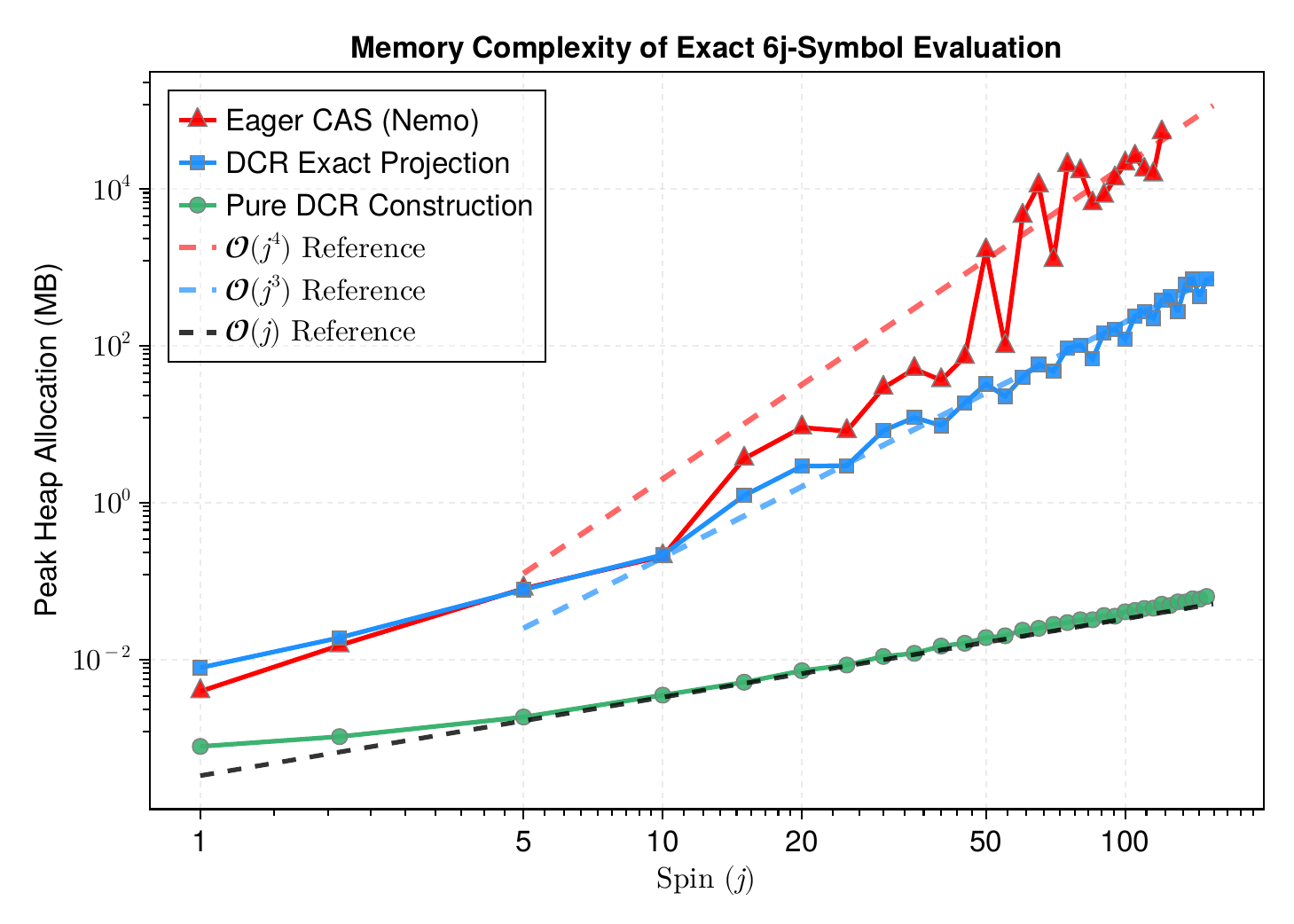}
\caption{Peak heap memory allocation for the exact algebraic evaluation of the symmetric $6j$-symbol. The DCR construction scales linearly with spin ($\mathcal{O}(j)$). Projecting the DCR into the exact field scales as $\mathcal{O}(j^3)$. In contrast, the eager CAS baseline triggers intermediate expression swell that scales worse than $\mathcal{O}(j^4)$, exceeding 50 GB by $j=120$.}
\label{fig:memory_benchmark}
\end{figure*}

A key limitation of exact $q$-hypergeometric evaluation is the rapid growth of intermediate symbolic expressions. To quantify this, we measure peak memory allocation for the symmetric $6j$-symbol at topological level $k = 4j$ across three regimes:
\begin{enumerate}
\item {\it DCR Construction:} Sparse integer exponent data $\{ {\cal M}_{\rm rad}, {\cal M}_{\rm root}, {\cal M}_{\rm base}, \{\mathcal R_z \} \}$. 
\item {\it Exact Projection:} Projection into a cyclotomic field, including the cost of DCR construction.
\item {\it Eager CAS Evaluation:} Direct rational hypergeometric summation with iterative GCD reductions.
\end{enumerate}

To ensure a rigorous comparison, the eager CAS baseline is highly optimized by evaluating the series iteratively using the exact rational hypergeometric update ratios ($\mathcal{R}_z$). At each step of the sum, the CAS multiplies the current rational term by $\mathcal{R}_z$ and immediately performs a polynomial greatest common divisor (GCD) reduction to maintain the expression in canonical lowest terms.

The peak allocations were measured using a double-pass execution strategy with forced garbage collection to eliminate JIT-compilation and LRU-caching artifacts. The memory snapshots are detailed in Table \ref{tab:memory_snapshots}, and the continuous scaling is plotted in Figure \ref{fig:memory_benchmark}.

\begin{table}[htpb]
\centering
\caption{Peak memory allocations for exact evaluation of the symmetric $6j$-symbol. The DCR remains compact (KB scale), exact projection grows predictably (remaining under 380 MB at $j=120$), while the eager CAS approach suffers intermediate expression swell (exceeds 50 GB at $j=120$).}
\label{tab:memory_snapshots}
\begin{tabular}{@{\hspace{5pt}}c@{\hspace{5pt}}cc@{\hspace{5pt}}c@{\hspace{5pt}}}
\toprule
\textbf{Spin $j$} & \textbf{DCR Build} & \textbf{Exact Projection } & \textbf{Eager CAS} \\
\midrule
20  & $7.2$ KB  & $2.89$ MB   & $9.07$ MB \\
40  & $14.8$ KB & $9.60$ MB   & $36.32$ MB \\
60  & $23.6$ KB & $39.74$ MB  & $4.52$ GB \\
75  & $29.5$ KB & $93.30$ MB  & $\bf 20.20 \; GB$ \\
120 & $50.6$ KB & $377.5$ MB & $\bf 52.47 \; GB$ \\
\bottomrule
\end{tabular}
\end{table}

\begin{table*}[htpb]
\centering
\caption{Evaluation of the symmetric $6j$-symbol at $k=500$. The eager LSE loses reliability beyond $j\gtrsim90$ due to catastrophic cancellation. DCR projection preserves correct signs and significantly reduces error in double precision.}
\label{tab:phase_flip}
\begin{tabular}{@{}r@{\hspace{15pt}}l@{\hspace{12pt}}l@{\hspace{12pt}}l@{}}
\toprule
\textbf{Spin $j$} & \textbf{Eager LSE (Float64)} & \textbf{DCR-Projection (Float64)} & \textbf{Truth (2048-bit)} \\
\midrule
30  & $-1.0930 \times 10^{-3}$ & $-1.0930 \times 10^{-3}$ & $-1.0930 \times 10^{-3}$ \\
50  & $\hphantom{-}9.1082 \times 10^{-4}$ & $\hphantom{-}9.1082 \times 10^{-4}$ & $\hphantom{-}9.1082 \times 10^{-4}$ \\
70  & $-7.6406 \times 10^{-4}$ & $-7.6286 \times 10^{-4}$ & $-7.6283 \times 10^{-4}$ \\
90  & $\mathbf{\hphantom{-}3.5642 \times 10^{-4}}$ & $\mathbf{-6.6327 \times 10^{-4}}$ & $\mathbf{-6.4428 \times 10^{-4}}$ \\
110 & $-9.6881 \times 10^{-1}$ & $\hphantom{-}1.5083 \times 10^{-3}$ & $\hphantom{-}2.8290 \times 10^{-4}$ \\
\bottomrule
\end{tabular}
\end{table*}

The key observations are 
\begin{itemize}
\item {\it DCR Construction:} Memory scales linearly with spin ($\mathcal{O}(j)$), costing under 51 kilobytes at $j=120$.
\item {\it Exact Projection:} Memory scales roughly cubically due to summation length, cyclotomic field dimension, and coefficient growth, consistent with Section~\ref{sec:complexity_error}. The observed oscillations are driven by the Euler totient function $\varphi(d)$, which dictates the exact degree of the underlying cyclotomic field.
\item {\it Eager CAS:} Polynomial GCD reductions cannot prevent expression swell, exceeding 50GB by $j=120$.
\end{itemize}

Overall, these results confirm the theoretical picture developed in Section~\ref{sec:complexity_error}. The DCR isolates combinatorial complexity from arithmetic growth, ensuring memory remains predictable and bounded prior to projection.

\begin{figure*}[htpb]
\centering
\includegraphics[width=0.97\linewidth]{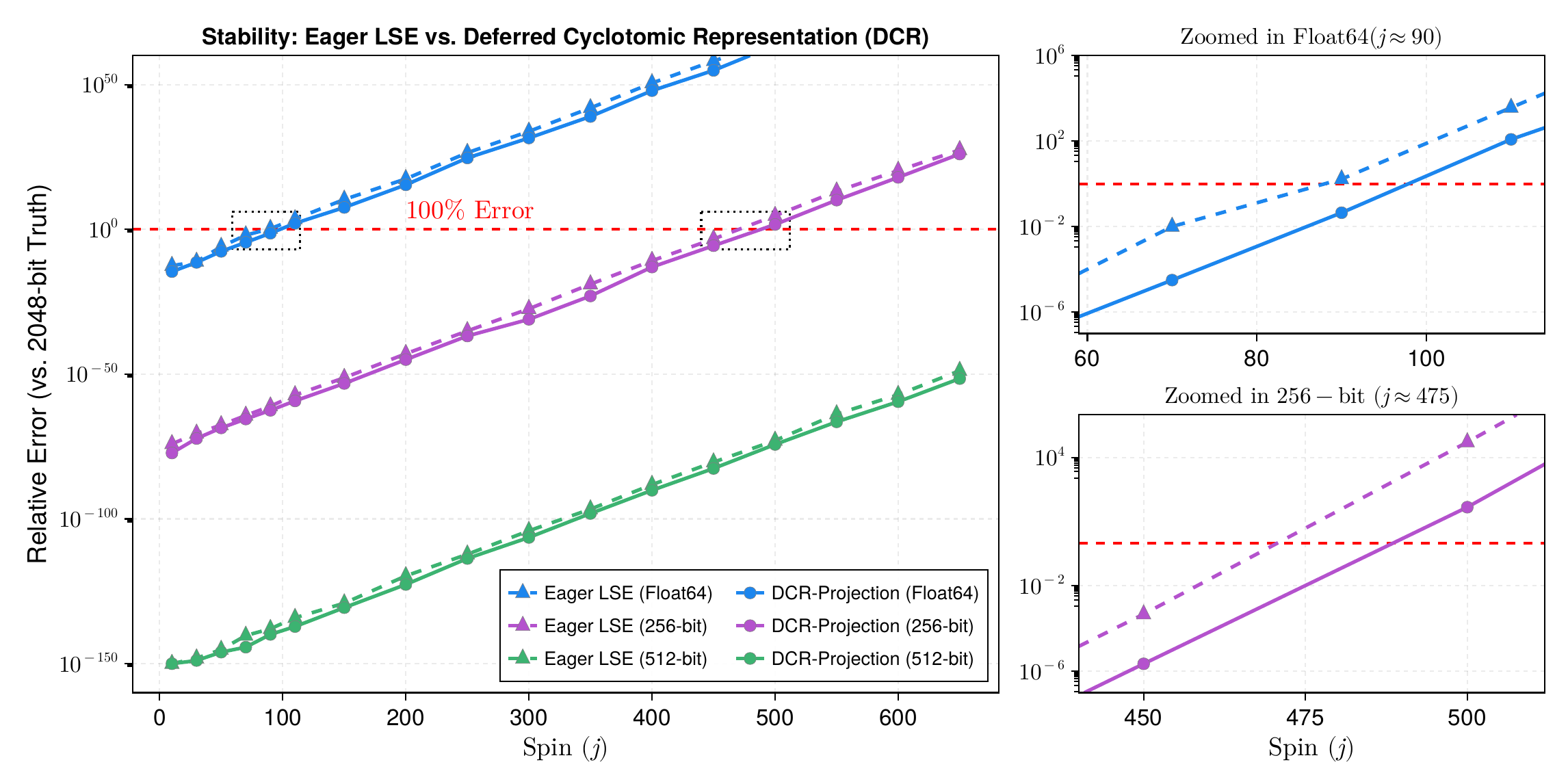}
\caption{Relative error of the symmetric $6j$-symbol at $k=10j$ for multiple precisions. Dashed lines denote the eager LSE baseline and solid lines the DCR projection. The red dashed line marks the 100\% relative error threshold. Right panel: Insets zoom into the onset of catastrophic cancellation. DCR systematically delays large errors and reduces their magnitude.}
\label{fig:stability_benchmark}
\end{figure*}

\subsection{Numerical Stability and Algebraic Preconditioning}
\label{subsec:stability_test}

We now empirically validate the stability properties established in Section~\ref{sec:complexity_error}. Specifically, we quantify the reduction of the representation-dependent amplification factor $\gamma_{\text{rep}}$ Eq.~\eqref{eq:gamma_rep} and its direct impact on finite-precision evaluation.

\medskip
\paragraph*{Dynamic Range Compression.} A key mechanism underlying the improved stability of the DCR is the exact algebraic cancellation of intermediate magnitudes prior to numerical evaluation. To quantify this preconditioning, Table~\ref{tab:dynamic_range} compares the peak dynamic range generated by eager polynomial evaluation ($\gamma_{\mathrm{eager}}$) against the symbolically reduced cyclotomic representation ($\gamma_{\mathrm{DCR}}$). We also track the intrinsic condition number ($\log_{10}\kappa$), which dictates the absolute theoretical minimum for precision loss.

\begin{table}[htpb]
\centering
\caption{Amplification factors for symmetric $6j$-symbols. Both $\gamma_{\mathrm{eager}}$ and $\gamma_{\mathrm{DCR}}$ are dimensionless log-scale measures of intermediate dynamic range. $\Delta\gamma$ represents the orders of magnitude compressed by symbolic preconditioning. The $\log_{10}\kappa$ column establishes the intrinsic precision loss.}
\label{tab:dynamic_range}
\begin{tabular}{@{}cc@{\hspace{12pt}}c@{\hspace{12pt}}c@{\hspace{12pt}}c@{\hspace{12pt}}c@{}}
\toprule
\textbf{Spin $j$} &  $k$ & \textbf{ $\log_{10}\kappa$} & $\boldsymbol{\gamma}_{\mathbf{eager}}$ & $\boldsymbol{\gamma}_{\mathbf{DCR}}$ & \textbf{ $\Delta\boldsymbol{\gamma}$} \\ 
\midrule
10  & 40   & 1.27  & 61.1   & 19.0   & 42.1   \\
50  & 200  & 7.10  & 560.1  & 104.5  & 455.6  \\
100 & 400  & 14.39 & 1352.7 & 212.5  & 1140.2 \\
200 & 800  & 28.97 & 3177.3 & 429.1  & 2748.2 \\
400 & 1600 & 58.12 & 7307.1 & 862.8  & 6444.3 \\
500 & 2000 & 72.70 & 9518.5 & 1079.8 & 8438.8 \\
\bottomrule
\end{tabular}
\end{table}

As illustrated in Table~\ref{tab:dynamic_range}, eager evaluation inflates intermediate expressions by thousands of decades. At $j=500$, the subtraction of unreduced polynomial components introduces an amplification factor of $\gamma_{\mathrm{eager}} = 9518.5$. The DCR algebraically eliminates this bloat at the exponent level, compressing the dynamic range by over $8400$ orders of magnitude ($\Delta\gamma$). This ensures the numerical processor only evaluates the irreducible core of the amplitude.

Although these quantities are not explicitly formed in floating-point arithmetic, they determine the scale of intermediate cancellations and therefore directly influence numerical error. As shown in Table~\ref{tab:dynamic_range}, the DCR consitently reduces intermediate magnitudes by several orders, leading to a substantial improvement in stability. This behavior is consistent with the reduction in $\gamma_{\mathrm{rep}}$ predicted in Section~\ref{sec:complexity_error}.

\begin{table*}[htpb]
\centering
\caption{Median steady-state execution latencies (milliseconds) for symmetric $6j$-symbols at $k=4j$. The DCR architecture isolates  combinatorial complexity into a one-time algebraic construction (Build), allowing subsequent numerical evaluations (Projection) to rival the speed of optimized eager LSE baselines.}
\label{tab:latency}
\begin{tabular}{@{}l@{\hspace{12pt}}l@{\hspace{12pt}}c@{\hspace{15pt}}cc@{\hspace{15pt}}cc@{}}
\toprule
& &  & \multicolumn{2}{c}{\textbf{Float64 (ms)}} & \multicolumn{2}{c}{\textbf{256-bit BigFloat (ms)}} \\
\cmidrule(lr){3-3} \cmidrule(lr){4-5} \cmidrule(lr){6-7}
\textbf{Spin $j$} & \textbf{Level $k$} & \textbf{DCR Build} & \textbf{Eager LSE} & \textbf{DCR Projection} & \textbf{Eager LSE} & \textbf{DCR Projection} \\
\midrule
50  & 200  & $0.0209$ & $0.0006$ & $0.0007$ & $0.1622$ & $0.2195$ \\
100 & 400  & $0.0824$ & $0.0009$ & $0.0014$ & $0.3163$ & $0.4607$ \\
200 & 800  & $0.3188$ & $0.0015$ & $0.0031$ & $0.6373$ & $0.9861$ \\
300 & 1200 & $0.6732$ & $0.0021$ & $0.0053$ & $0.9530$ & $1.5305$ \\
400 & 1600 & $1.1638$ & $0.0027$ & $0.0075$ & $1.2716$ & $2.0571$ \\
500 & 2000 & $1.7698$ & $0.0031$ & $0.0094$ & $1.6052$ & $2.6260$ \\
\bottomrule
\end{tabular}
\end{table*}

\medskip
\paragraph*{Finite-Precision Accuracy.}
To assess the practical impact of this preconditioning, we evaluate the symmetric $6j$-symbol at fixed level ($k=500$) in double precision and compare against a 2048-bit reference (Table~\ref{tab:phase_flip}). The standard LSE implementation begins to exhibit significant deviation at moderate spins and produces sign errors by $j \approx 90$, indicating loss of numerical reliability. In contrast, the DCR-based evaluation maintains the correct sign and substantially improved accuracy over a wider range.

This behavior is further illustrated in Figure~\ref{fig:stability_benchmark}, where the relative error of the eager LSE implementation rapidly approaches unity, while the DCR maintains controlled error growth.  The residual error at large $j$ is governed by the intrinsic condition number $\kappa$, confirming that the DCR has exhausted the representation-level improvement. Further gains require either higher precision arithmetic or a fundamentally different summation order.

\subsection{Execution Latency and Amortized Parameter Sweeps}

We conclude our empirical analysis by examining execution time. Runtimes were measured using Julia's \texttt{BenchmarkTools}, reporting the median steady-state latencies over warmed-up repeated executions to bypass JIT-compilation overhead.

We compare the DCR projection against an eager numerical baseline that applies the Log-Sum-Exp (LSE) algorithm directly to precomputed logarithmic $q$-factorial tables, reducing the summation to rapid sequential array access. While both engines ultimately utilize LSE to safely accumulate the final alternating sum, the DCR engine pre-conditions the calculation. By algebraically canceling the massive factorials into sparse cyclotomic ratios prior to numerical evaluation, the DCR compresses the dynamic range of the intermediate terms. The resulting stability thresholds across hardware and arbitrary-precision regimes are summarized in Table~\ref{tab:latency}.

At standard 64-bit precision, the direct eager LSE evaluation achieves very low latency. However, as demonstrated in Section~\ref{subsec:stability_test}, its numerical reliability deteriorates rapidly with increasing spin due to catastrophic cancellation, leading to 100\% loss of accuracy beyond $j \approx 90$ (see Figure~\ref{fig:stability_benchmark}). The algebraic preconditioning inherent in the DCR delays this instability, extending the range of reliable double-precision evaluation to moderately larger spins. Nevertheless, for sufficiently large $j$, both approaches require promotion to software-emulated arbitrary precision to maintain accuracy. In this regime, the amortized advantages of the deferred representation become increasingly significant, as the cost of projection grows more slowly than direct high-precision summation.

In the 256-bit regime, both methods exhibit comparable scaling. The optimized LSE baseline benefits from direct accumulation over precomputed high-precision tables, yielding $1.61$ ms at $j=500$. The DCR projection achieves similar performance ($2.63$ ms at $j=500$), indicating that traversal and evaluation of the sparse cyclotomic representation introduces only modest overhead relative to direct (raw tables) summation. This is consistent with the complexity analysis, where both approaches ultimately reduce to arithmetic operations in the target field.

The primary computational advantage of the deferred representation emerges in repeated evaluations across parameter space. In direct methods, each new value of the deformation parameter $q$ requires recomputation of the full summation, including reconstruction of auxiliary tables. In contrast, the DCR decouples combinatorial structure from numerical evaluation. The structure is constructed once, and subsequent evaluations require only the projection step. For moderate spins within the hardware precision limit (e.g., $j \leq 100$), this reduces marginal evaluation cost to the microsecond regime ($0.7\,\mu\mathrm{s}$ at $j=50$, $1.4\,\mu\mathrm{s}$ at $j=100$). This enables efficient exploration of continuous parameter regimes that would otherwise require repeated full recomputation.

\medskip

In summary, the latency results align with the theoretical framework developed in Section~\ref{sec:complexity_error}. The deferred cyclotomic representation isolates combinatorial complexity into a compact, reusable structure, while the remaining computational cost is governed by arithmetic in the target field. This separation does not eliminate the intrinsic cost of high-precision evaluation, but it avoids redundant recomputation and enables efficient amortization across repeated evaluations.

\section{Conclusion and Outlook}
\label{sec:conclusion}

We have introduced the deferred cyclotomic representation (DCR), a structural reformulation for the evaluation of finite $q$-hypergeometric series and $q$-deformed amplitudes. Instead of treating quantum amplitudes as a dense polynomial representation, the DCR encodes their multiplicative structure in a sparse integer exponent representation over cyclotomic factors. Evaluation in any target field is then realized via projection as a ring homomorphism applied to this canonical combinatorial object.

The computational advantages of this framework are structural rather than merely algorithmic. By executing exact factorial cancellations at the exponent level prior to projection, the DCR functions as a universal algebraic preconditioner. For exact symbolic computation, it eliminates the expression swell inherent in rational polynomial representations. For numerical evaluation, it substantially compresses the intermediate dynamic range, systematically delaying the onset of catastrophic cancellation. Furthermore, as established in our macroscopic analysis (see Section~\ref{subsec:state_sums}), compiling the DCR isolates and amortizes the combinatorial complexity across entire triangulations. This separation of structure and evaluation renders large-scale topological state sums computationally tractable.

The behavior of the DCR in the semiclassical regime deserves particular emphasis. While large-spin and high-level $k$ limits induce high oscillations and severe intrinsic cancellation, the DCR systematically suppresses the additional numerical amplification caused by conventional summand representations. By performing exact cyclotomic cancellations prior to numerical evaluation, the method significantly extends the range of parameters computable in finite precision. Consequently, the dominant limitation in the local amplitude evaluation is no longer algebraic expression swell, but the fundamental oscillatory nature of the amplitude itself. In this sense the DCR should be viewed as an enabling representation, isolating the intrinsic numerical barriers and opening a substantially larger computational window to quantitatively test Regge-type asymptotic formulas against exact quantum amplitudes. 

Beyond its computational advantages, the DCR provides a new structural perspective on $q$-deformed amplitudes. The deformation parameter $q$ does not enter the internal representation, but appears only through evaluation. As a consequence, amplitudes at different deformation parameters arise as evaluations of a single underlying combinatorial object. Within this framework, admissibility at roots of unity and the classical limit $q \to 1$ are naturally expressed in terms of cyclotomic exponent support, rather than as external analytical constraints. This provides a direct link between quantum and classical regimes at the level of factorization structure.

Several directions for further investigation emerge from this work. On the structural side, it is natural to ask whether coherence relations in quantum recoupling theory, such as orthogonality and the Biedenharn--Elliott identity, admit a formulation directly at the level of exponent data prior to evaluation. On the algorithmic side, extending the DCR to higher-rank quantum groups and more general tensor network contractions \cite{Dittrich_2016,Asante:2024eft} will require maintaining sparsity under increasingly complex combinatorial operations. 

More broadly, this work highlights the role of representation design in computational mathematics and physics. By aligning the representation with the intrinsic multiplicative structure of the problem, it is possible to simultaneously improve exact computability, numerical stability, and algorithmic efficiency. The deferred cyclotomic representation provides one concrete realization of this principle in the context of $q$-hypergeometric series, and suggests that similar approaches may be applicable to broader classes of special functions and quantum amplitudes.

Ultimately, the deferred cyclotomic framework establishes that representation design prior to evaluation can fundamentally alter the computational profile of topological quantum field theories. By aligning the data structure with the true multiplicative algebra of the deformation, we obtain a representation that is simultaneously more computable, numerically stable, and theoretically transparent.

\section*{Code Availability}
The methods developed in this work are implemented natively in the open-source Julia package \href{https://github.com/sethkasante/QRecoupling.jl}{\texttt{QRecoupling.jl}} \cite{QRecoupling_pkg}. The package extends the deferred cyclotomic architecture beyond the $6j$-symbol to encompass quantum $3j$-symbols, general fusion tensors, and braiding $R$-matrices, and general $q$-hypergeometric series. It provides a unified, high-performance interface for exact algebraic evaluation, root-of-unity projection, and arbitrary-precision floating-point computation. The repository includes a \texttt{benchmarks/} directory with scripts that reproduce the tables and figures in this paper. 

\section*{Acknowledgments}
The author is grateful for the funding and support of the Atlantic Association for Research in the Mathematical Sciences (AARMS). I would also like to thank Bianca Dittrich, Sebastian Steinhaus, and Edward Wilson-Ewing for their support and encouragement throughout this work. This work was supported in part by the Natural Sciences and Engineering Research Council of Canada.

\bibliography{refs}

\end{document}